\documentclass[notnumberedtheorems,withtitlethanks]{actacyb}

\usepackage{graphicx, algorithm, algorithmic}
\usepackage{bbm}
\usepackage{amsmath} 
\newcommand{\indic}{\mathbbm{1}}
\newcommand{\jecko}{\hat{\jmath}}
\newcommand{\jedna}{1}
\newcommand{\real}{\mathbbm{R}}
\DeclareMathOperator*{\argmin}{\arg\!\min}
\DeclareMathOperator*{\argmax}{\arg\!\max}
\newcommand{\tuci}{{\cal I}}
\newtheorem{model}{Model} 
\usepackage{rotating}

\hypersetup{
    unicode=true,            
    pdftitle={Title},        
    pdfauthor={Authors},     
    pdfsubject={Manuscript submitted to Acta Cybernetica},   
    colorlinks=true,         
    linkcolor=blue,          
    citecolor=blue,          
    filecolor=blue,          
    urlcolor=blue            
}

\begin{document}



\title{A Bayesian Framework for Regularized Estimation in Multivariate Models Integrating Approximate Computing Concepts}
\headingtitle{Bayesian Framework for Regularized Multivariate Estimation}

\author{
Jan Kalina\thanks{The Czech Academy of Sciences, Institute of Computer Science, Prague, Czech Republic} \thanks{\email{kalina@cs.cas.cz}, \orcid{0000-0002-8491-0364}} 
}
\headingauthor{Jan Kalina}


\maketitle

\begin{abstract}
This paper discusses regularized estimators in the multivariate statistical model as tools naturally arising within a Bayesian framework. 
First, a link is established between Bayesian estimation and inference under parameter rounding (quantization), thereby connecting two distinct paradigms: Bayesian inference and approximate computing.
Next, Bayesian estimation of the means from two independent multivariate normal samples is employed to justify shrinkage estimators, i.e., means shrunk toward the pooled mean.
Finally, regularized linear discriminant analysis (LDA) is considered. Various shrinkage strategies for the mean are justified from a Bayesian perspective, and novel algorithms for their computation are proposed.
The proposed methods are illustrated by numerical experiments on real and simulated data.

\end{abstract}

\section{Introduction} 

In a~variety of models of multivariate statistics and machine learning, regularized estimators may be obtained in the Bayesian statistical setup \cite{cal}.
This paper is interested in finding interpretations and relationships between various types of estimators in the multivariate model including an extension for a~classification model with data observed in several different groups. 

Regularized (shrinkage) estimators are not only crucial for high-dimensional data analysis, where classical methods often become unstable or infeasible, but also beneficial even when working with small datasets, as they tend to improve predictive accuracy \cite{val}. In a wide range of probabilistic models, regularization emerges naturally from Bayesian reasoning: prior knowledge is incorporated directly into the estimation process, resulting in estimates that are pulled (shrunk) towards a~predefined reference value, typically the prior expectation~\cite{hos}.

A shrinkage estimator is a statistical estimator that improves estimation accuracy by combining observed data with additional information or assumptions, effectively "shrinking" raw estimates toward a predetermined target or prior value. This shrinkage reduces estimation variance, often at the cost of introducing some bias, but typically results in a lower overall mean squared error compared to unbiased estimators.

This form of regularization helps reduce variance and improve generalization performance, which is especially valuable when the data are noisy. In multivariate settings, such Bayesian-inspired regularization enhances the stability and robustness of the estimators \cite{joh}. We regard this as a natural and fruitful synthesis of Bayesian ideas and modern regularization techniques, aligning well with the goals of reliable and interpretable data analysis \cite{has}.

As a novelty, a connection between Bayesian estimators and quantized estimation, i.e.~estimation contaminated by additional uncertainty or error, is established here. Moreover, we focus on shrinkage estimators for the means of independent multivariate normal samples and subsequently on regularized linear discriminant analysis, which frequently employs shrinkage techniques and has found many applications in high-dimensional data \cite{fu}. In particular, we justify the use of shrunken means within a specific Bayesian model and formulate two algorithms for regularized linear discriminant analysis exploiting linear algebra manipulations.

This paper has the following structure. Section~\ref{sec:mult} is devoted to Bayesian estimation of parameters of the multivariate normal distribution. 
In Section~\ref{sec:rounding}, an~original connection between Bayesian estimation and estimation in a~model with rounding (quantization) of the estimated parameters is established.
In Section~\ref{sec:two}, Bayesian estimation for two independent normally distributed samples is derived, which theoretically justifies using regularized estimators in the given context and also in the context of
regularized linear discriminant analysis of Section~\ref{sec:lda}. 
Section~\ref{sec:ex} presents experiments over a real and a simulated dataset. Section~\ref{sec:con} concludes the paper.

\section{Bayesian estimation in the multivariate model} 
\label{sec:mult}

Let us assume the multivariate model 
\begin{equation}
X_i \sim {\sf N}_p(\mu,\Sigma), \quad i=1,\dots,n,
\label{e:mult}
\end{equation}
with i.i.d.~$p$-variate random vectors following the $p$-variate normal distribution with $n>p$.
Let $X\in\real^{n \times p}$ denote the data matrix with rows $X_1,\dots,X_n$.
The covariance matrix $\Sigma$ has to fulfil $\Sigma \in {\sf PD(p)}$, where ${\sf PD}(p)$ denotes the set of all positive definite symmetric matrices of size $p \times p$.

This section discusses Bayesian estimation of the mean vector \(\mu \in \real^p\) and the covariance matrix \(\Sigma\). The main purpose here is to review the Bayesian estimation of these parameters, serving as a foundation for the subsequent sections. For a~detailed description of Bayesian estimation principles, we refer the reader to the book \cite{joh}.

The sample mean has the form
\begin{equation}
\bar{X}=(\bar{X}_1,\dots,\bar{X}_p)^T  = \frac{1}{n} X^T \jecko_n  = \frac{1}{n}\left(  \jecko_n^TX \right)^T \in \real^p, \quad \mbox{where} \quad \bar{X}_j = \frac{1}{n} \sum_{i=1}^n X_{ij}
\label{e:mean}
\end{equation}
and $\jecko_n=(1,\dots,1)^T\in\real^n$. 
Also, the maximum likelihood estimator in (\ref{e:mult}) is the sample mean. In addition, 
minimizing the sum of squared residuals
\begin{equation}
\argmin_{\tilde{\mu}\in\real^p} \sum_{i=1}^n (X_i - \tilde{\mu})^2
\end{equation}
leads again to the sample mean as the resulting estimator.

Under the assumption of a~known~$\Sigma$, let us consider the normal prior
\begin{equation}
\mu \sim {\sf N}_p(\theta, \eta)
\end{equation}
with known $\theta\in\real^p$ and $\eta\in{\sf PD}(p)$. 
This is a conjugate prior for the multivariate normal likelihood with known covariance, meaning that the posterior distribution belongs to the same parametric family (i.e.~multivariate normal).
The use of a~conjugate prior is advantageous because it leads to a closed-form expression for the posterior, which simplifies both analytical derivations and computational implementation of Bayesian inference.
The Bayesian estimator $\hat{\mu}$ of the parameter~$\mu$ obtained as the mean of the posterior distribution is a shrinkage estimator of $\mu$ shrunken towards~$\theta$. 
It depends on the nuisance covariance matrix~$\Sigma$. The formula for the estimator, which was provided e.g.~in the seminal paper \cite{eva}, is given by
\begin{equation}
\hat{\mu} =(n\Sigma^{-1}+\eta^{-1})^{-1} (n \Sigma^{-1} \bar{X}+ \eta^{-1} \theta).
\label{e:slunce}
\end{equation}

\begin{remark}
If assuming $\eta^{-1}=c \Sigma^{-1}$ for some $c>0$, the estimator (\ref{e:slunce}) does not depend on $\Sigma$, because (\ref{e:slunce}) simplifies to
\begin{equation}
\hat{\mu}=(1-\delta)\bar{X} + \delta \theta, \quad\mbox{where}\quad \delta=\frac{c}{n+c}.
\label{e:bayspe}
\end{equation}
\end{remark}

The James--Stein estimator is a biased estimator of $\mu$ with a smaller quadratic risk than the mean \cite{jam} for $p \geq 3$.
Let us assume a single $X \sim {\sf N}_p(\mu, \sigma^2\tuci_p)$ with a~known $\sigma^2>0$, where $\tuci_p$ denotes the identity matrix of dimension $p$. The estimator is formulated as an estimator shrunken towards 0 in the form
\begin{equation}
\left(1-\frac{(p-2)\sigma^2}{||X||^2_2}\right)X,
\label{e:js}
\end{equation}
which has the structure of the Bayesian estimator in (\ref{e:bayspe}) obtained with a Gaussian prior with a zero expectation.

The James--Stein estimator (\ref{e:js}) is a foundational result demonstrating the advantages of shrinkage estimation under the assumption of a known diagonal covariance matrix with equal variances $\sigma^2$ across all coordinates. Although these assumptions may not hold in all practical scenarios, the estimator’s core principles have inspired a wide range of effective regularization methods that are widely used today.

An analogous shrinkage estimator can also be defined for the sample mean of a Gaussian distribution when the covariance matrix is diagonal and known, but the coordinate-wise variances $\sigma_1^2, \ldots, \sigma_p^2$ may differ \cite{efr}. This extension relaxes the assumption of homoscedasticity and allows for heterogeneous variability across coordinates. The classical James--Stein estimator nevertheless remains a cornerstone theoretical result, providing strong justification for the use of regularized (i.e. Bayesian) estimators in multivariate settings \cite{jam}.

Building on the concept of shrinkage in (\ref{e:js}), many popular regularized estimators have been developed, including ridge regression and lasso regression. These ideas also form the basis of modern regularization techniques frequently employed in machine learning, such as bagging, boosting, and weight decay in deep learning~\cite{has}.
In practice, when parameter variability such as variance is unknown, Empirical Bayes methods offer a useful alternative to the Stein estimator. They estimate hyperparameters directly from the data and enable adaptive, data-driven shrinkage~\cite{riz}.

If the task is to estimate the covariance matrix~$\Sigma$, it is common—especially in high-dimensional applications—to apply Tikhonov regularization, replacing the empirical covariance matrix~$S$ by
\begin{equation}
\tilde{S} = (1-\lambda)S+\lambda T \in {\sf PD}(p)
\label{e:shrsigma}
\end{equation}
with a given regular target matrix~$T\in{\sf PD}(p)$, depending on a tuning parameter $\lambda \in [0,1)$. This regularized estimator~$\tilde{S}$ is guaranteed to be positive definite and well-conditioned, even when $S$ is singular or ill-conditioned due to a limited sample size or high dimensionality. This not only ensures the existence and uniqueness of subsequent solutions (e.g.~in discriminant analysis), but also reduces the sensitivity of the model to noise, improves numerical stability, and mitigates overfitting.

Moreover, $\tilde{S}$ can also be interpreted from a Bayesian perspective \cite{cal}, where the target matrix~$T$ corresponds to the covariance matrix of the prior distribution of the random~$\Sigma$. The asymptotically optimal value of $\lambda$ minimizing the mean squared error in~(\ref{e:mult}) was derived in~\cite{led} for the case $T = \tuci_p$, eliminating the need for cross-validation. This result was later extended to more general choices of~$T$ in~\cite{sch}. Regularized estimators of~$\Sigma$ that are additionally robust to outliers have been proposed for data following elliptically symmetric unimodal distributions; see e.g.~\cite{gsc,bou}.
In practice, we often encounter situations where parameters are not estimated precisely but, for example, rounded (quantized). This leads us to consider how quantization relates to the Bayesian perspective.


\section{Connection between quantization and Bayesian estimation} 
\label{sec:rounding}

In this section, we formalize the connection between quantized estimation—that is, estimation involving rounding or discretization—and Bayesian estimation. Quantization introduces bias while potentially reducing variability, and it can be interpreted as a projection of the estimator onto a discrete set. Therefore, it is of interest to investigate whether quantization can be regarded as an implicit Bayesian estimator.

One common form of quantization, known as post-training rounding, has become a widely used strategy for reducing energy consumption and computational latency in statistical and machine learning models \cite{nah}. The key idea is to replace high-precision parameter estimates by their low-precision (typically integer or fixed-point) approximations while maintaining acceptable predictive performance.

Model comparison via reduced numerical precision falls under approximate computing, which studies how controlled approximations (such as rounding, noise, or perturbations) affect model behavior and reliability.
In deep learning, such approximations are deliberately introduced to model small errors in learned parameters, aiming to simplify deployment on resource-constrained devices \cite{sze}. Similarly, when training is performed using low-precision arithmetic, rounding effects become inherent to the learning process and must be accounted for in both the design and evaluation of 
models~\cite{sah}.

Quantization is also gaining traction in probabilistic modeling and Bayesian inference. For example, a Bayesian regression model with quantized parameters was proposed in \cite{yan} for meteorological and hydrological applications, where the goal was to quantify how uncertainties in precipitation and temperature propagate through parameter estimates and affect predicted runoff in a snowmelt-driven watershed. More recently, numerical experiments have demonstrated how parameter quantization impacts inference in Bayesian 
networks~\cite{rib}.

We consider a multivariate model, which is a special case of~(\ref{e:mult}).
If the distribution of the error (e.g., due to rounding) is known, estimating~$\mu \in \real^p$ becomes a~Bayesian estimation problem.
The error distribution then plays the role of a~prior on~$\mu$ in the original, uncontaminated model.
The mathematical results highlight a~conceptual link between approximate computing and Bayesian inference.

\begin{model} 
Let us consider the model 
\begin{equation}
X_i = \mu + e_i, \quad i=1,\dots,n,\label{e:true1}
\end{equation}
with $X_i \in \real^p$, a fixed $\mu\in\real^p$, and i.i.d. random vectors $e_1,\dots,e_n$, which follow $e_i \sim {\sf N}_p(0, \sigma^2 \tuci_p)$ for $i=1,\dots,n$. Instead of the true model (\ref{e:true1}), the parameter estimation will be rather considered in the model 
\begin{equation}
X_i = \xi + \varepsilon_i, \quad i=1,\dots,n, \quad \mbox{where}~\xi~\mbox{is obtained as}~\xi=\mu+\tau
\end{equation}
with a random vector $\tau \sim {\sf N}_p(0,\delta^2 \tuci_p)$ independent of $\varepsilon_1,\dots,\varepsilon_n$.
\label{model1}
\end{model}

\begin{model} 
Let us consider the model 
\begin{equation}
X_i = \xi+e_i, \quad i=1,\dots,n,
\end{equation}
where i.i.d. random vectors $e_1,\dots,e_n$ follow $e_i \sim {\sf N}_p(0, \sigma^2 \tuci_p)$ for $i=1,\dots,n$ and $\xi\in\real^p$ is a random vector following $\xi \sim {\sf N}_p(\mu, \delta^2 \tuci_p)$ with a given $\mu\in\real^p$.
\label{model2}
\end{model}

\begin{theorem}
The posterior estimator of $\xi$ in Model~\ref{model1} is the same as the posterior estimator of $\xi$ in Model~\ref{model2}.
\label{th:1}
\end{theorem}

\begin{proof}
The idea is that a random vector $\xi$ defined as $\xi = \mu + \tau$,  
where $\mu$ is a~constant vector and $\tau \sim \mathcal{N}(0, \delta^2 I_p)$,  
satisfies $\xi \sim \mathcal{N}_p(\mu, \delta^2 I_p)$.
\end{proof}

\vspace{3mm}
It is now natural to ask what happens if Model~\ref{model1} considers a prior distribution about $\mu$ instead of taking a fixed $\mu$.
In such a setup with a double uncertainty (about $\mu$ and due to the quantization), the next theorem explains that standard Bayesian estimation may be exploited combining both sources of uncertainty in an~additive way.

\begin{model} 
Let us consider the model 
\begin{equation}
X_i = \mu + e_i, \quad i=1,\dots,n,
\label{e:true2}
\end{equation}
with $X_i \in \real^p$, a random vector $\mu \sim {\sf N}_p(\theta,\Psi)$ with a fixed $\theta \in \real^p$, $\Psi \in {\sf PD}(p)$, and i.i.d. random vectors $e_1,\dots,e_n$ following $e_i \sim {\sf N}_p(0, \sigma^2 \tuci_p)$ for $i=1,\dots,n$.
Instead of the true model (\ref{e:true2}), the parameter estimation will be rather considered in the model 
\begin{equation}
X_i = \xi + \varepsilon_i, \quad i=1,\dots,n, \quad \mbox{where}~\xi~\mbox{is obtained as}~\xi=\mu+\tau
\end{equation}
with a random vector $\tau \sim {\sf N}_p(0, \delta^2 \tuci_p)$ independent of $\mu$ and $\varepsilon_1,\dots,\varepsilon_n$.
\label{model3}
\end{model}

\begin{model} 
Let us consider the model 
\begin{equation}
X_i = \xi+e_i, \quad i=1,\dots,n,
\end{equation}
where i.i.d. random vectors $e_1,\dots,e_n$ follow $e_i \sim {\sf N}_p(0, \sigma^2 \tuci_p)$ for $i=1,\dots,n$ and $\xi\in\real^p$ is a random vector following $\xi \sim {\sf N}_p(\theta, \Psi+\delta^2 \tuci_p)$ with a given $\theta\in\real^p$ and $\Psi \in {\sf PD}(p)$.
\label{model4}
\end{model}

\begin{theorem}
The posterior estimator of $\xi$ in Model~\ref{model3} is the same as the posterior estimator of $\xi$ in Model~\ref{model4}.
\end{theorem}

\begin{proof}
Let us assume $\mu \sim {\sf N}_p(\theta,\Psi)$ and $\tau \sim {\sf N}_p(0,\delta^2\tuci_p)$ independent of $\mu$. In this case, the idea is that $\xi$ defined as $\xi=\mu+\tau$ fulfils $\xi \sim {\sf N}_p(\theta, \Psi+\delta^2 \tuci_p)$.
The posterior mean of $\xi$ is the Bayes estimator under squared error loss.
\end{proof}

The connection between quantization and Bayesian statistics represents a significant shift in perspective. First, it marks a paradigm change by moving from viewing quantization as a purely deterministic, technical artifact to understanding it as an inherent part of probabilistic modeling. Second, it offers a unifying framework that bridges seemingly distinct concepts—quantization, rounding errors, and Bayesian inference—under a common statistical umbrella. Finally, this insight has practical implications: it suggests new ways to design algorithms that explicitly incorporate uncertainty introduced by quantization, potentially leading to more robust and interpretable models in real-world applications.

A limitation of our presented approach lies in the fact that quantization is modeled as deterministic post-training rounding, without considering its interaction with training dynamics or data structure; alternatively, rounding errors could be modeled stochastically, for example as random perturbations drawn from a~uniform distribution. Next, we turn to applying Bayesian principles to a specific problem: the shrinkage estimation of group means. Shrinkage can be viewed as a form of regularization, which is naturally interpreted as a Bayesian estimate with a suitable prior.

\section{Shrunken means as Bayesian estimators for two groups}
\label{sec:two}

Assuming two independent random samples, it is common in high-dimensional applications to consider their means shrunken towards a constant value (e.g., the pooled mean). For instance, shrinking both the means and the covariance matrix jointly was studied in \cite{ber} in the context of hypothesis testing for data from multiple groups; however, the study did not compare the tests based on shrunken means with those using standard (non-shrunken) means. Other examples include the classification task of \cite{guo} or more recent classifiers based on various specialized norms~\cite{sif,li}.

In high-dimensional settings, the standard estimator of the mean often becomes highly unstable. Applying shrinkage to group means can dramatically reduce the mean squared error (MSE). Understanding the connection between shrinkage and Bayesian estimation enables the systematic derivation of shrinkage rules, such as those employed in linear discriminant analysis (LDA), which will be discussed in detail in the next section (Section~\ref{sec:lda}).

An important and common scenario involves applying shrinkage across two or more estimators, thereby effectively borrowing information across different samples or groups. While shrinkage estimators are often presented in the literature as empirically motivated techniques, their Bayesian interpretation is rarely made explicit. In this work, we provide a theoretical justification for these existing shrinkage methods, which frequently appear as ad hoc solutions or are employed without clearly acknowledging their underlying Bayesian connection.

This section focuses on justifying the use of shrinkage means within a specific Bayesian framework, assuming the covariance matrix is regular.
First, let us formulate the result for univariate data and then we extend it to the multivariate case.

\begin{lemma}
Let us assume two independent random samples 
\begin{equation}
X_1,\dots,X_n \sim {\sf N}(\mu_x, \sigma^2) \quad \mbox{and}\quad Y_1,\dots,Y_m \sim {\sf N}(\mu_y, \sigma^2)
\end{equation}
of univariate observations with a common variance $\sigma>0$. 
Let us assume that the prior distributions for $\mu_x$ and~$\mu_y$ are independent and of the form
\begin{equation}
\mu_x \sim {\sf N}(\theta, \gamma^2) \quad\mbox{and}\quad \mu_y \sim {\sf N}(\theta, \gamma^2)
\label{e:twoprior}
\end{equation}
with $\theta \in \real^p$ and $\gamma>0.$ Then, estimates of $\mu_x$ and $\mu_y$ obtained as the means of the posterior distribution have the form
\begin{equation}
\hat{\mu}_x = \frac{n \bar{X} \sigma^{-2} + \theta \gamma^{-2}}{n \sigma^{-2} + \gamma^{-2}} \quad\mbox{and}\quad \hat{\mu}_y = \frac{m \bar{Y} \sigma^{-2} + \theta \gamma^{-2}}{m \sigma^{-2} + \gamma^{-2}}.
\label{e:twouni}
\end{equation}
\end{lemma}

\begin{proof}
The likelihood of the data across the two random samples, up to a multiplicative constant, is 
\begin{equation}
f(\mbox{data}|\mu_x,\mu_y,\sigma^2) \propto \exp \left\{-\frac{1}{\sigma^2}  \left[ \sum_{i=1}^n (X_i-\mu_x)^2 + \sum_{j=1}^m (Y_j-\mu_y)^2 \right] \right\}.
\end{equation}
The prior distribution, if assuming the distribution of $\mu_x$ to be independent from that of $\mu_y$, corresponds to two one-dimensional densities with the product
\begin{equation}
f(\mu_x)f(\mu_y) = \left( \frac{1}{2 \pi \gamma^2} \right)^2 \exp \left\{ -\frac{1}{\gamma^2} \left[ (\mu_x-\theta)^2 + (\mu_y-\theta)^2 \right] \right\}.
\end{equation}
The posterior density can be expressed as $f(\mu_x,\mu_y|\mbox{data}) \propto$ 
\begin{equation}
\exp \left\{ - \frac{\sum_{i=1}^n (X_i-\mu_x)^2}{\sigma^2}- \frac{\sum_{j=1}^m (Y_j-\mu_y)^2}{\sigma^2}- \frac{(\mu_x-\theta)^2}{\gamma^2}- \frac{(\mu_y-\theta)^2}{\gamma^2}\right\}.
\end{equation}
Let us search for the partial derivatives of the posterior density. Using the chain rule, 
the derivatives of the inside function will be put equal to 0. 
Solving this set of two equations for $\mu_x$ and $\mu_y$ leads us finally to 
\begin{equation}
\hat{\mu}_x = \frac{\gamma^2 \sum_{i=1}^n X_i + \theta \sigma^2}{n \gamma^2 + \sigma^2} \quad\mbox{and}\quad \hat{\mu}_y = \frac{\gamma^2 \sum_{j=1}^m Y_j + \theta \sigma^2}{m \gamma^2 + \sigma^2},
\end{equation}
which can be equivalently reformulated as (\ref{e:two}), i.e.~the proof is concluded. 
\label{l:two}
\end{proof}

\begin{lemma}
Let us assume two independent random samples 
\begin{equation}
X_1,\dots,X_n \sim {\sf N}_p(\mu_x, \Sigma) \quad \mbox{and}\quad Y_1,\dots,Y_m \sim {\sf N}_p(\mu_y, \Sigma)
\end{equation}
of $p$-variate observations with $p>1$. We assume a~common covariance matrix $\Sigma\in{\sf PD}(p)$. 
Let us assume that the prior distributions for $\mu_x$ and~$\mu_y$ are independent and of the form
\begin{equation}
\mu_x \sim {\sf N}_p(\theta, \Upsilon) \quad\mbox{and}\quad \mu_y \sim {\sf N}_p(\theta, \Upsilon)
\end{equation}
with $\theta \in \real^p$ and $\Upsilon \in {\sf PD}(p)$. Then, estimates of $\mu_x$ and $\mu_y$ obtained as the means of the posterior distribution have the form
\begin{align}
\hat{\mu}_x &= (n \Sigma^{-1} + \Upsilon^{-1})^{-1} ( n \Sigma^{-1} \bar{X} + \Upsilon^{-1} \theta), \nonumber\\
\hat{\mu}_y &= (m \Sigma^{-1} + \Upsilon^{-1})^{-1} ( m \Sigma^{-1} \bar{Y} + \Upsilon^{-1} \theta).
\label{e:two}
\end{align}
\label{l:mult}
\end{lemma}

The multivariate estimates obtained in Lemma~\ref{l:mult} are shrunken to the common value $\theta,$ in an analogy to (\ref{e:twouni}). Particularly, (\ref{e:two}) may be reformulated in the elegant form as
\begin{equation}
\hat{\mu}_x = (\tuci_p-\Delta_x) \bar{X} + \Delta \theta \quad\mbox{and}\quad \hat{\mu}_y = (\tuci_p-\Delta_y) \bar{Y} + \delta \theta, 
\label{e:ele}
\end{equation}
where the matrices $\Delta_x \in  \real^{p \times p}$ and $\Delta_y \in \real^{p \times p}$ have the form
\begin{equation}
\Delta_x = (n \Sigma^{-1}+\Upsilon^{-1})^{-1} \Upsilon^{-1} \quad\mbox{and}\quad \Delta_y = (m \Sigma^{-1}+\Upsilon^{-1})^{-1} \Upsilon^{-1}.
\end{equation}
These Bayesian estimators of $\mu_x$ and $\mu_y$ depend on $\Sigma$, just like the estimator~(\ref{e:slunce}) in Section~\ref{sec:mult}.
In the specific situation with $\Upsilon^{-1}=c \Sigma^{-1}$ with $c>0$ and $n=m$, the estimators retain the form (\ref{e:ele}) while $\Delta_x$ and $\Delta_y$ simplify to a common matrix
\begin{equation}
\Delta_x = \Delta_y = \frac{c}{n+c} \tuci_p.
\end{equation}
For practical applications, it is natural to take $\theta$ in (\ref{e:twoprior}) to be the pooled mean of the data across both groups.
In other words, we justify here the commonly used approach based on shrinking the means in the context of LDA. The results of this section can be simply extended to the case with more than two groups.
On the other hand, relying on specific distributional assumptions and the availability of a~regular covariance matrix may limit the applicability of the method; furthermore, the performance and robustness of shrunken means under model misspecification remain beyond the scope of this section.

\section{Regularized estimators within linear discriminant analysis}
\label{sec:lda}

Regularized versions of LDA are studied in this section, along with novel algorithms for their computation. These methods leverage the shrinkage estimation of means discussed in Section~\ref{sec:two}, where shrinkage of the means represents one part of the overall regularization strategy. Another crucial component is the regularization of the covariance matrix. This technique is primarily applied to ensure numerical stability and invertibility, especially in high-dimensional settings where the covariance matrix may be ill-conditioned or singular. Together, these regularization approaches enable more stable and reliable classification.
In fact, LDA tends to perform poorly in high-dimensional settings; incorporating regularization enhances both its generalization capability and numerical stability, making it better suited for such challenging scenarios.

Linear discriminant analysis (LDA) is a classification method for $p$-variate observations
\begin{equation}
X_{11},\dots,X_{1n_1}, \dots, X_{K1},\dots,X_{Kn_K}
\label{e:kgroups}
\end{equation}
observed in $K$ different groups with $p > K \geq 2$. We assume the data in each of the $k=1,\dots,K$ groups to follow the normal ${\sf N}_p(\mu_k,\Sigma)$ distribution and
allow the data to be high-dimensional allowing $\min_{k=1,\dots,K} n_k <p$. 
If the data are high-dimensional, which is a common situation e.g.~in molecular genetics, econometrics or chemistry, using a regularized estimate of $\Sigma$ is typically highly beneficial~\cite{bos}. While the means of each group are commonly computed as arithmetic means, their regularization might be beneficial as well. Shrinking the means towards the pooled mean across groups was used e.g.~in~\cite{fu}, however without pointing out at the Bayesian connections discussed above.

Regularized estimation of $\Sigma$ will be considered here as a~replacement of the pooled empirical covariance matrix $S$ by the shrinkage version~$\tilde{S}$ 
(\ref{e:shrsigma}) depending on $\lambda \in [0,1]$. 
The Bayesian connections of $\tilde{S}$ were discussed already in Section~\ref{sec:mult}.
Let us remark that $\tilde{S}$ evaluates the covariance structure around $\bar{X}$; we are not aware of any alternative approach considering the covariance structure around a regularized mean. 

Let us now formulate several possible versions of regularized means inspired by Section~\ref{sec:two}, which correspond to different forms of regularization.
They depend on regularization parameters and it is natural to select the values that lead to the best classification performance; such values can be found in a cross-validation.
We use the notation $\bar{X}$ for the pooled mean, $\bar{X}_k$ for the mean of the $k$-th group, and $\indic$~for indicator function.

\begin{definition}[Regularized means for $k=1,\dots,K$]
\label{def:mean}
\begin{enumerate}\mbox{}\\[0.1ex]
\item
  \begin{equation}
  \bar{X}_k^{(2)} = (1-\delta^{(2)}) \bar{X}_k + \delta^{(2)} \bar{X}, \quad \delta^{(2)}\in(0,1).
	\label{mean:2}
  \end{equation}
\item
  \begin{eqnarray}
  \bar{X}_{k}^{(1)} 
  &=& {\sf sgn}(\bar{X}_k) \max\left\{ |\bar{X}_k| - \delta^{(1)}, 0\right\} \nonumber\\ 
	&=& {\sf sgn}(\bar{X}_k) \left( |\bar{X}_k|-\delta^{(1)}\right)_+, ~~\delta^{(1)}\in(0,1).
  \label{e:soft}
  \end{eqnarray}
\item
  \begin{eqnarray}
	\bar{X}_k^{(0)}
  &= &\bar{X}_k \indic\left[ |\bar{X}_k| > \delta^{(0)} \right]\nonumber\\
	&= &\bar{X}_k \indic\left[ \bar{X}_k \notin (-\delta^{(0)}, \delta^{(0)})\right]\nonumber\\
	&= &\left\{ \begin{array}{lll}
         \bar{X}_k, && |\bar{X}_k| \geq \delta^{(0)},\\
				 0, && |\bar{X}_k| \leq \delta^{(0)}.\end{array} \right.
	\label{ehard}
	\end{eqnarray}
\end{enumerate}
\end{definition}


Both (\ref{e:soft}) and (\ref{ehard}) may yield sparse solutions, which is particularly appealing in classification tasks involving high-dimensional data, as illustrated in image analysis applications~\cite{kalmat}. In particular, hard thresholding (\ref{ehard}) may result in some coordinates being set to zero; if a given coordinate is zero across all groups~$k$, the corresponding variable effectively drops out—this constitutes a form of feature selection. The hard thresholding in (\ref{ehard}) is inspired by image denoising applications of \cite{dav}, where the goal is to remove noise while preserving important features by setting small coefficients to zero.

None of these estimators is universally optimal. Regularization based on the $L_2$-norm often outperforms $L_1$-based methods in various statistical models~\cite{pio}, as it tends to work well when all variables contribute meaningfully to the prediction. Conversely, $L_1$-based approaches are more appropriate when only a small number of dominant variables are expected to be relevant. The hard thresholding in~(\ref{ehard}) takes a more aggressive sparsity-inducing stance by directly setting small coefficients to zero rather than continuously shrinking them, which can be particularly effective when a clear zero/non-zero distinction is desired. Moreover, in situations where continuous shrinkage of average values compromises interpretability, the hard thresholding (\ref{ehard}) provides an appealing alternative by enforcing explicit sparsity through hard cutoffs.

The Bayesian interpretation of the soft-thresholding estimator~(\ref{e:soft}),  
in which a~Laplace prior induces shrinkage, provides a probabilistic justification for its use.  
In contrast, to our knowledge, no analogous Bayesian framework has been established for  
the hard-thresholding estimator~(\ref{ehard}), suggesting an interesting direction for future research.

To our knowledge, no version of LDA using~(\ref{ehard}) has been proposed in the literature; we introduce it here as a simplified alternative to~(\ref{e:soft}), in which no shrinkage is applied when $\bar{X}_k \geq \delta^{(0)}$.
The version of LDA, where only $S$ is regularized but not the means, will be denoted here simply as RLDA. 
Versions with regularization applied also on the means will be denoted as RLDA(0), RLDA(1), and RLDA(2) according to the regularization type applied on the means according to Definition~\ref{def:mean}. 
We consider $Z \in \real^p$ as an observation to be classified. The classification rule assigns $Z$ to the group that is obtained by 
\begin{equation}
\tilde{k} := \argmax_{k=1,\dots,K} l_k,
\label{e:optk}
\end{equation}
where $l_k$ is the discriminant score. 
Prior knowledge about group membership will be denoted as $\pi_k$ for the group $k=1,\dots,K$, where $\sum_{k=1}^K \pi_k=1$. 
The formal definition follows.

\begin{definition}[Discriminant scores of regularized LDA for $k=1,\dots,K$]
\label{def:lda}
\begin{enumerate}\mbox{}\\[0.1ex]
\item Plain RLDA
  \begin{equation}
  \tilde{\ell}_k = \bar{X}_k^T \tilde{S}^{-1} {Z}  - \frac{1}{2} \bar{X}_k^T \tilde{S}^{-1} \bar{X}_k + \log \pi_k.
	\label{lda1}
  \end{equation}
\item RLDA$(\omega)$ for $\omega\in \{0,1,2\}$
  \begin{equation}
  \tilde{\ell}_k^{(\omega)} = \left(\bar{X}_k^{(\omega)}\right)^T \tilde{S}^{-1} Z - \frac{1}{2} \left(\bar{X}_k^{(\omega)}\right)^T \tilde{S}^{-1} \bar{X}_k^{(\omega)} + \log \pi_k.
	\label{e:optlda}
  \end{equation}	
\end{enumerate}
\end{definition}

\begin{algorithm}[t] 
\caption{RLDA(2) for a general $T$ based on Cholesky decomposition.}
\label{alg:1}
\begin{algorithmic}[1]
\REQUIRE Data (\ref{e:kgroups}) and $Z \in \real^p$.
\REQUIRE $T \in {\sf PD}(p)$.
\REQUIRE $\delta \in [0,1]$ and $\lambda \in [0,1]$.
\REQUIRE $\pi_1,\dots,\pi_K$ with $\sum_{k=1}^K \pi_k=1$.
\ENSURE Classification decision for given $Z \in \real^p$.
\STATE Compute $\tilde{S}$ according to 
 \begin{equation}
  \tilde{S} = (1-\lambda)S + \lambda T.
	\end{equation}
\STATE Compute the Cholesky factor $L_\lambda$ of $\tilde{S}$ as
  \begin{equation}
 \tilde{S} = L_\lambda L_\lambda^T,
	\end{equation}
	where $L_\lambda$ is a nonsingular lower triangular matrix.
\STATE Compute the matrix 
  \begin{equation}
  B^{\delta \lambda} = L_\lambda^{-T} \left[(1-\delta)\bar{X}_1+\delta \bar{X} - Z,\,\dots\,,(1-\delta)\bar{X}_K+\delta \bar{X} - Z\right],
  \end{equation}
	where $L_\lambda ^{-T}= (L_\lambda^{-1})^T = (L_\lambda^T)^{-1}$,
	and assign $Z$ to group $k$, if 
	\begin{equation}
	k = \argmax_{j=1,\dots,K} \left\{ ||B_j^{\delta \lambda}||^2_2 + \log \pi_k\right\},
	\end{equation}
	where $||B_j^{\delta \lambda}||^2_2$ is the squared Euclidean norm of the $j$-th column of $B^{\delta \lambda}$.
\end{algorithmic}
\end{algorithm}

A direct computation of the discriminant scores in (\ref{lda1}) and (\ref{e:optlda}) is highly inefficient due to the need to invert the empirical covariance matrix. Therefore, we present two efficient algorithms for computing regularized LDA.
Both use fixed regularization parameters $\delta$ and $\lambda$ that require to be optimized within cross-validation.
Its advantage is also avoiding the need to specify the unknown~$\Sigma$, which serves as a nuisance parameter here, and to specify hyperparameters (i.e.~parameters of the underlying prior).
Algorithm~\ref{alg:1} is based on Cholesky decomposition of $\tilde{S}$. The optimal $k$ in (\ref{e:optlda}) exists with probability 1.

Let us also consider a~specific form of regularization, namely the ridge regularization (as denoted e.g.~in \cite{has2}) in the form
\begin{equation}
S^* = \lambda S + (1-\lambda) \tuci_p, \quad \lambda \in [0,1].
\label{esi}
\end{equation}
We use the notation 
\begin{equation}
\tilde{X}  = [X_{11}-\bar{X},\dots,X_{1n_1}-\bar{X}, \dots, X_{K1}-\bar{X},\dots,X_{Kn_K}-\bar{X}]^T \in \real^{n \times p}
\label{eY}
\end{equation} 
and express the pooled estimator $S$ in the form $S = \tilde{X}^T \tilde{X}$. 

Algorithm~\ref{alg:1}, which is based on the Cholesky decomposition, is efficient and stable for well-conditioned covariance matrices but may fail or become unstable when the matrix is nearly singular or not positive definite. In contrast, Algorithm~\ref{alg:2}, relying on singular value decomposition, provides greater numerical stability and robustness in such challenging cases, albeit at a higher computational cost. Despite this, the numerical results produced by both algorithms are expected to be highly comparable. Both methods have similar asymptotic computational complexity of order \( O(p^3) \), where \( p \) denotes the size of the covariance matrix. Nevertheless, Cholesky decomposition is typically 2 to 4 times faster than SVD in practical runtime measurements~\cite{whi}.

\begin{algorithm}[t]
\caption{RLDA(2) for the ridge regularization.}
\label{alg:2}
\begin{algorithmic}[1]
\REQUIRE Data (\ref{e:kgroups}) with $n<p$ and $Z \in \real^p$.
\REQUIRE $\delta \in [0,1]$ and $\lambda \in [0,1]$.
\REQUIRE $\pi_1,\dots,\pi_K$ with $\sum_{k=1}^K \pi_k=1$.
\ENSURE Classification decision for given $Z \in \real^p$.
\STATE Compute the matrix $\tilde{X}$ as in (\ref{eY}).
\STATE Compute the singular value decomposition of $\tilde{X}$ as
  \begin{equation}
  \tilde{X}= P\Omega Q^{T},
  \end{equation}
  with singular values $\{ \omega_1,\dots,\omega_n \}$ and complement these singular values with $p-n$ zero values $\omega_{n+1}=\dots=\omega_{p}=0.$ 
\STATE 
       Estimate the variances of the columns of $\tilde{X}$, denoted by $\hat{\sigma}_1^2, \dots, \hat{\sigma}_p^2$
\STATE For a fixed $\lambda \in [0,1],$ compute
  \begin{equation}
  D_\lambda = {\sf diag} \{\lambda\hat{\sigma}_1^2 + (1-\lambda),\dots, \lambda\hat{\sigma}_p^2 + (1-\lambda)\},
	\label{edstar}
  \end{equation}
	where ${\sf diag}$ denotes a diagonal matrix.
\STATE Compute the matrix 
  \begin{equation}
	B^{\delta \lambda} = D_\lambda^{-1/2}Q^{T} \left[(1-\delta)\bar{X}_1+\delta \bar{X} - Z,\,\dots\,,(1-\delta)\bar{X}_K+\delta \bar{X}  - Z\right] 
	\end{equation}
	and assign $Z$ to group $k$, if 
	\begin{equation}
	k = \argmax_{j=1,\dots,K} \left\{ ||B_j^{\delta \lambda}||^2_2 + \log \pi_k\right\}.
	\end{equation}
\end{algorithmic}
\end{algorithm}

Slight modifications can be used to adapt the algorithms for other LDA versions of Definition~\ref{def:lda}.
All the described versions of LDA, which are extensions of simple regularized LDA versions of \cite{kaljur}, are sensitive to outliers, because the mean and empirical covariance matrix are vulnerable to the presence of outliers in the data. A desirable robustification can be performed by exploiting a suitable highly robust estimator, such as the minimum weighted covariance determinant (MWCD) estimator~\cite{roe,kalbea}. 
Moreover, regularization enables the extension of LDA to high-dimensional settings, including the incorporation of robust estimators such as those described in \cite{hub}. A limitation of this approach is that the decision rule remains linear, relying on the assumption of normally distributed classes with a~common covariance matrix; this may reduce classification performance when group distributions are non-Gaussian or exhibit different covariance structures.


\section{Experiments}
\label{sec:ex}

The aim of the experiments is to study the regularized LDA of Section~\ref{sec:lda} on real as well as simulated data.

As a real dataset, we use the genetic data coming from the cardiovascular genetic study performed in the Center of Biomedical Informatics in Prague. 
The data are gene expressions for 38\,590 gene transcripts measured for 46 patients with acute myocardial infarction (AMI) and 49 control persons, previously analyzed in \cite{kalbea}, 

In a simulation study, we randomly generate two independent random samples both with $n=m=50$ observations and $p=1000$ variables.
The first random sample is generated from ${\sf N}_p(0, \Sigma)$ with $\Sigma = \sigma^2 [(1-c)\tuci_p + c \jedna_p \jedna_p^T]$,
where $\sigma=1$, $c=0.4$, and $\jedna_p = (1,\dots,1)\in\real^p$. The second random samples is generated from ${\sf N}_p(\mu_y, \Sigma)$,
where $\mu_y= (3,3,3,3,3, 0,\dots,0)^T\in\real^p$ is the vector with the first 5~coordinates being non-zero.

The methods used to analyze the two datasets are listed in Table~\ref{tab:ex}.
Because standard LDA cannot be computed due to a singular covariance matrix, we use regularized versions using the shrinkage targets
$T_1=\tuci_p$ or $T_2=\sigma^2 \tuci_p + \vartheta (\jedna_p \jedna_p^T - \tuci_p)$ with $\vartheta=0.15$.
All our implementations use the Cholesky decomposition of the covariance matrix.
If the means are also regularized, then the regularization types of Definition~\ref{def:mean} are used.
For finding the regularization parameters, 5-fold cross-validation is typically used.
If only the covariance matrix is regularized, we also use the asymptotically optimal values of the regularization parameter; this was derived for $T_1$ as well as~$T_2$ in \cite{led2003}.
For comparisons, prediction analysis of microarrays (PAM) was computed with default regularization for the means; no regularization of the covariance matrix is needed due its diagonal structure \cite{tib}.

The results are presented in Table~\ref{tab:ex}. PAM has a weaker classification performance than any of the regularized LDA versions.
The best result is obtained in both datasets when $T_2$ is used as a shrinkage target and when doing the hard thresholding for the means.
Results obtained with the shrinkage target $T_1$, which is the most common type, stay only behind.
Regularizing the means improves the performance in every case compared to situations with standard means.

The choice of regularized LDA with $T_2$ and $L_1$-regularization for the means exactly corresponds to SCRDA of \cite{guo} implemented in the {\tt rda} package of {\sf R} software~\cite{rda}; the best results obtained with $T_2$ and hard thresholding for the means overcome those 
of \cite{guo} remarkably. 
We can say that hard thresholding is most effective when the chosen regularization corresponds to a strong prior belief and the threshold is suitably selected.
In the genetic dataset, it turns out that a set of a few hundred of variables is responsible for the classification between the two groups; such finding corresponds to the analysis of \cite{kalbea}. In the simulation, the most successful LDA versions exploit only 5 variables,
which corresponds to the true number of variables where the shift between the two groups occurs.

It is also worth noting that the shrinkage target $T_2$ precisely corresponds to the data-generating design used in the simulation;
for the genetic data, $T_2$ is apparently a reasonable target as well, as the variables exhibit non-negligible correlations, making a diagonal covariance structure inappropriate.

The asymptotic regularization parameter values turn out to be overcome by those found in the cross-validation; this is because the approach of Ledoit-Wolf~\cite{led} does not consider regularized means and the optimality of the regularization parameter holds only asymptotically and would be relevant only for larger numbers of observations. We may conclude that not only the regularization as such, but also the way of finding the parameters are crucial for the classification performance.

\begin{sidewaystable} 
\centering
\begin{tabular}{|c|ccc|ccr|ccr|}
\hline
         &             &             &                 & \multicolumn{3}{|c|}{Genetic data} & \multicolumn{3}{c|}{Simulated data}\\
Classif. & Cov.~matrix & Expectation & Choice of       &           &      & \# of     &          &      & \# of \\
method   & $T$         & Reg.~type   & reg. parameters & Accuracy  & (SD) & variables & Accuracy & (SD) & variables\\
\hline
PAM & None & $L_1$ & CV &0.69&(0.05)&541   &0.75&(0.05)&51\\
LDA & $T_1$ & - & CV     &0.77&(0.03)&38\,590   &0.84&(0.03)&1000\\
LDA & $T_1$ & - & LW     &0.74&(0.03)&38\,590   &0.82&(0.03)&1000\\
LDA & $T_2$ & - & CV     &0.79&(0.03)&38\,590   &0.86&(0.03)&1000\\
LDA & $T_2$ & - & LW     &0.76&(0.03)&38\,590   &0.84&(0.03)&1000\\
LDA & $T_1$ & $L_2$ & CV &0.78&(0.02)&38\,590   &0.86&(0.03)&1000\\
LDA & $T_1$ & $L_1$ & CV &0.82&(0.02)&287       &0.88&(0.03)&21\\
LDA & $T_1$ & (\ref{ehard}) & CV  &0.83&(0.02)&260       &0.88&(0.03)&14\\
LDA & $T_2$ & $L_2$ & CV &0.81&(0.02)&38\,590   &0.90&(0.03)&1000\\
LDA & $T_2$ & $L_1$ & CV &0.83&(0.02)&222       &0.91&(0.03)&5\\
LDA & $T_2$ & (\ref{ehard}) & CV  &0.86&(0.02)&213       &0.92&(0.03)&5\\
\hline
\end{tabular}
\caption{Results of the experiments of Section~\ref{sec:ex}. The classification accuracy, its standard deviation evaluated in 5-fold cross-validation, and the number of selected variables used in the classification rule are reported.
Methods not performing variable selection use as many as $p$ variables. Different methods are compared, using the regularization parameters obtained in 5-fold cross-validation (CV) or using the Ledoit-Wolf asymptotics (LW). 
}
\label{tab:ex}
\end{sidewaystable}

\section{Conclusion}
\label{sec:con}

Regularized (shrinkage) estimators have already been exploited in a variety of multivariate applications. Meanwhile, regularized versions of LDA represent popular tools for the analysis of high-dimensional data.
Importantly, the theoretical results and proposed estimators are derived for finite sample sizes and do not rely on asymptotic approximations, which makes them directly applicable in practical scenarios.
This paper justifies some regularized estimators by deriving them within specific Bayesian contexts. 
At the same time, regularization is known to be related to local robustness, i.e.~to the known resistance of regularized estimators to local shift modification of the majority of the data~\cite{hof,xie}, which can be explained within the framework of robust optimization \cite{xan}.

The proposed methods were illustrated by numerical experiments on both real and simulated data, supporting their practical relevance.
Experiments not shown here verified that the results for Algorithm~\ref{alg:2} are very close to those of Algorithm~\ref{alg:1}, if the ridge regularized is used.

The Bayesian approach is explained in Section~\ref{sec:rounding} to allow also a different and original interpretation: it corresponds to training the parameters perturbed by errors. As a consequence, 
researchers interested in the effect of post-training rounding of the parameters or computing in a low-precision arithmetic should focus on Bayesian (i.e.~regularized) estimation.

The Bayesian thinking used here may also be applied to a~number of other multivariate contexts. 
Because regularized LDA is non-robust to outliers, we plan to perform future research of novel robust regularized classification methods based on the highly robust MWCD estimator \cite{roe}.

Among the main limitations of the presented approach are (1) the need to specify a prior distribution, which is an inherent feature of any Bayesian framework, and (2) the sensitivity to model misspecification, a drawback stemming from the parametric nature of the model rather than from Bayesian inference itself. Nevertheless, in practice, cross-validation can be used to select regularization parameters that correspond to suitable prior choices, thus mitigating the impact of subjective prior specification.
In addition, a certain degree of robustness can be achieved in robust LDA by using robust estimators of the mean and covariance matrix~\cite{hub2024}.

\section*{Acknowledgement}

The research was supported by the grant 25-15490S of the Czech Science Foundation.
The author would like to thank Michal Jel\'\i{}nek and Anton Matis, who were supported by the program Strategy AV21 ``AI: Artificial Intelligence for Science and Society''.
We thank the Institute of Computer Science of the Czech Academy of Sciences for its support and congratulate it on 50 years of excellence.

\bibliographystyle{actaplain}
\bibliography{kalina}
\end{document}